\documentclass[12pt]{amsart}
\usepackage[margin=1in]{geometry}
\usepackage{amssymb,amsmath,amsthm,bbold,mathtools,bm} 
\def\-{\raisebox{.75pt}{-}}
\usepackage{enumitem}
\usepackage{comment}
\usepackage{tikz} 
\usetikzlibrary{cd}
\usetikzlibrary{arrows}

\numberwithin{table}{section}
\numberwithin{equation}{section}

\newtheorem{defn}{Definition}[section]

\newtheorem{rmk}{Remark}[section]
\newtheorem{lem}{Lemma}[section]
\newtheorem{prop}{Proposition}[section]
\newtheorem{thm}{Theorem}[section]
\newtheorem*{thm*}{Theorem}

\theoremstyle{plain}

\newcommand{\fh}{\mathfrak{h}}

\newcommand{\fl}{\mathfrak{l}}

\newcommand{\sE}{\mathcal{E}}
\newcommand{\sF}{\mathcal{F}}

\newcommand{\sH}{\mathcal{H}}

\newcommand{\R}{\mathbb{R}}
\newcommand{\C}{\mathbb{C}}
\newcommand{\Hq}{\mathbb{H}}
\newcommand{\Z}{\mathbb{Z}}

\newcommand{\Q}{\mathbb{Q}}

\newcommand{\U}{\mathrm{U}}

\newcommand{\rad}{\mathrm{rad}}
\newcommand{\Tor}{\mathrm{Tor}}

\newcommand{\vol}{\mathrm{vol}}
\newcommand{\covol}{\mathrm{covol}}

\newcommand{\Ber}{\mathrm{Ber}}

\usepackage{hyperref,cleveref}
\usepackage[alphabetic,initials]{amsrefs}

\title{Higher abelian gauge theory in BV formalism}
\author{Shuhan Jiang}
\address{
	Department of Mathematics, University of Zurich, Winterthurerstrasse 190, CH-8057 Zurich, Switzerland
}
\email{shuhan.jiang@math.uzh.ch}

\begin{document}
	\begin{abstract}
		A BV structure is constructed on the space of zero modes of a quadratic elliptic higher abelian gauge theory. The associated BV integral is shown to be independent of the choice of zero-mode Lagrangian under mild hypotheses. The construction yields explicit zero-mode partition functions for several $p$-form abelian gauge theories.
	\end{abstract}
	\maketitle
	
	\section{Introduction}

In this paper, we study the zero-mode partition function of a quadratic $p$-form abelian gauge theory on a closed oriented Riemannian manifold $(X,g)$. A global description of such a theory must retain not only the differential-form fields visible perturbatively, but also their integral and torsion sectors. We therefore formulate the classical fields in Cheeger--Simons differential cohomology \cite{CheegerSimons}, with field space $\check H^{p+1}(X)$ and a quadratic action
\[
S:\check H^{p+1}(X)\longrightarrow \C/i\Z.
\]

Differential cohomology has been used extensively to formulate higher abelian gauge theories globally \cite{BelovMoore, FreedMooreSegal,Szabo2012,BeckerBeniniSchenkelSzabo,MooreSaxena}. A Batalin--Vilkovisky (BV) formulation of such a theory can be described using complexes of sheaves; see \cite{ElliottGwilliamSaberiWilliams}. We briefly recall the relevant construction. A cochain complex modeling $\check H^{p+1}(X)$ is the $(p+1)$-st smooth Deligne complex
\[
\mathrm{Bund}_\nabla(p+1)
\coloneqq
\Z[p+1]
\longrightarrow
\Omega^0[p]
\xlongrightarrow{d}\cdots
\xlongrightarrow{d}
\Omega^p;
\]
see \cite{Brylinski}. Its cotangent complex is represented by
\[
\Omega_{\mathrm{cl}}^{n-p}
\coloneqq
\Omega^{n-p}
\xlongrightarrow{d}\cdots
\xlongrightarrow{d}
\Omega^n[-p].
\]
For a local quadratic action, the equation-of-motion map is induced by a morphism of complexes
\[
K:\mathrm{Bund}_\nabla(p+1)
\longrightarrow
\Omega_{\mathrm{cl}}^{n-p}.
\]
The BV complex is the $(-1)$-shifted mapping cone
\[
\sF\coloneqq\operatorname{Cone}(K)[-1],
\]
and its hypercohomology $\mathbb H^\bullet(X;\sF)$ is the space of zero-modes of the theory.

The hypercohomology $\mathbb H^\bullet(X;\sF)$ is naturally a collection of abelian Lie groups. Its Lie algebra records the perturbative zero modes, while its exponential lattices and component groups retain the topological data. The canonical shifted symplectic structure and the Riemannian metric induce compatible perturbative BV data, and the classical action restricts to a quadratic weight on the degree-zero component group. We package these structures in the notion of an \emph{abelian BV group}.

For a maximal admissible Lagrangian $L\subset\mathbb H^\bullet(X;\sF)$, we define a finite-dimensional BV integral $Z_{\mathbb H(X;\sF)}(L)$. Our main result is the following.
\begin{thm*}
	Under mild hypotheses, the zero-mode partition function $Z_{\mathbb H(X;\sF)}(L)$ is independent of the choice of maximal admissible Lagrangian $L$.
\end{thm*}

Explicitly, the zero-mode partition function consists of a graded volume of connected zero-modes, a finite torsion factor, and a theta sum over degree-zero topological sectors. It captures the global finite-dimensional contribution to the partition function, separately from the regularized determinants of nonzero modes familiar from quadratic functionals and analytic torsion \cite{RaySinger1971,Schwarz1978}.

We work out the construction for Maxwell, abelian Chern--Simons, and Maxwell--Chern--Simons theories. In particular, for $p$-form Maxwell theory, our zero-mode formula is similar in form to that of \cite{MooreSaxena}, but differs in its torsion factor. In a companion paper \cite{JiangLiZabzine}, the formulas developed here are used in a BV-theoretic analysis of abelian duality.

\subsection*{Acknowledgments}

The author thanks Owen Gwilliam for valuable comments and suggestions that improved the manuscript, and Vivek Saxena for a helpful discussion.

The author acknowledges partial support of the SNF Grant No. 200021 227719 and of the Simons Collaboration on Global Categorical Symmetries. This research was (partly) supported by the NCCR SwissMAP, funded by the Swiss National Science Foundation. This article is based upon work from COST Action 21109 CaLISTA, supported by COST (European Cooperation in Science and Technology) (www.cost.eu), MSCA-2021-SE-01-101086123 CaLIGOLA, and MSCA-DN CaLiForNIA-101119552.

\section{Abelian BV groups}

In this section, we introduce an abstract BV structure on a collection of abelian Lie groups, modeled on the global zero modes of higher abelian gauge theories.

\subsection{Preliminaries}

Let $V = \bigoplus_{k\in \Z} V^k$ be a finite-dimensional graded vector space.

\begin{defn}
	A \emph{$(-1)$-shifted symplectic pairing} on $V$ is a nondegenerate graded skew-symmetric pairing
	\[
	\omega \colon V \times V \longrightarrow \R[-1].
	\]
	
	Graded skew-symmetry means that $\omega(V^k,V^\ell)=0$ unless $k+\ell=1$, and
	\[
	\omega(x,y)=-(-1)^{(k-1)(\ell-1)}\omega(y,x).
	\]
	for $x \in V^k$ and $y \in V^\ell$. 
	
	Nondegeneracy means that the map
	\[
	\omega^\flat \colon V \xlongrightarrow{\sim} V^\vee[-1],
	\qquad x\longmapsto\omega(x,-),
	\]
	is an isomorphism.
\end{defn}

For a graded subspace $W\subset V$, let
\[
W^{\perp_\omega} \coloneq \{x\in V\mid \omega(x,W)=0\}.
\]
We call $W$ \emph{isotropic}, \emph{coisotropic}, or \emph{Lagrangian} if, respectively,
\[
W\subset W^{\perp_\omega},
\qquad
W^{\perp_\omega}\subset W,
\qquad
W=W^{\perp_\omega}.
\]

\begin{defn}
	A (graded) lattice $\Gamma \subset V$ is called \emph{integral} if
	\[
	\omega(\Gamma,\Gamma)\subset\Z.
	\]
	It is called \emph{full} if
	\[
	\Gamma_\R \coloneqq  \Gamma\otimes_{\Z}\R \cong V.
	\]
    Its \emph{radical} is
	\[
	\rad(\Gamma)\coloneq\{\gamma\in\Gamma\mid \omega(\gamma,\Gamma)=0\}.
	\]
\end{defn}

An integral lattice is \emph{coisotropic} if $\Gamma_\R$ is coisotropic. In this case, $\rad(\Gamma)$ is a full lattice in $\Gamma_\R^{\perp_\omega}$, and $\omega$ induces a nondegenerate integral pairing $\underline\omega$ on
\[
\underline\Gamma \coloneq \Gamma/\rad(\Gamma).
\]
We call a coisotropic lattice $\Gamma$ \emph{reduced unimodular} if
\[
\underline \omega^\flat:
\underline\Gamma\xlongrightarrow{\sim}
\underline\Gamma^\vee[-1]
\]
is an isomorphism.

Let $A$ be a finitely generated abelian group. 

\begin{defn}
A complex \emph{quadratic weight} on $A$ is a map
\[
w:A\longrightarrow\C/i\Z
\]
such that
\begin{itemize}
	\item $w(x+y)-w(x)-w(y)$ is biadditive;
	\item $\Re\bigl(w(x)\bigr)=\Re\bigl(w(-x)\bigr)$.
\end{itemize}
\end{defn}

Because the period group $i \Z \subset \C$ is purely imaginary, the function
\[
q(x)\coloneq\Re\bigl(w(x)\bigr)
\]
is a well-defined real quadratic weight on $A$.

\begin{lem}\label{lem:real-quadratic}
The function $q$ vanishes on $\Tor(A)$ and descends uniquely to a real quadratic weight
\[
\bar q:A_{\mathrm{free}}\coloneq A/\Tor(A)\longrightarrow\R
\]
satisfying $\bar q(nx)=n^2\bar q(x)$ for every $n\in\Z$.
\end{lem}

\begin{proof}
Let $b(x,y)\coloneqq q(x+y)-q(x)-q(y)$. Since $q(-x)=q(x)$, the identity $0=q(x-x)=2q(x)-b(x,x)$ gives $b(x,x)=2q(x)$. Since $2 \in \R$ is invertible, $q(nx)= b(nx,nx)/2=n^2q(x)$. If $t$ is torsion and $nt=0$, then $0=q(nt)=n^2q(t)$, so $q(t)=0$. Moreover $n b(x,t)=b(x,nt)=0$, and therefore $b(x,t)=0$. Thus $q(x+t)=q(x)$, proving that $q$ descends to $A_{\mathrm{free}}$.
\end{proof}

\subsection{BV structures}

Let $H=\{H^k\}_{k\in\Z}$ be a collection of abelian Lie groups, trivial in all but finitely many degrees, such that each Lie algebra $\fh^k$ of $H^k$ is finite-dimensional and each component group $\pi_0(H^k)$ is finitely generated.

Write
\[
\fh \coloneq \bigoplus_{k\in\Z}\fh^k,
\qquad
\Gamma^k \coloneq \ker\bigl(\exp:\fh^k\to H^k\bigr),
\qquad
\Gamma \coloneq \bigoplus_k\Gamma^k.
\]
For every $k$ there is a canonical exact sequence
\begin{equation}\label{eq:exp-sequence}
	0\longrightarrow \Gamma^k
	\longrightarrow \fh^k
	\xlongrightarrow{\exp} H^k
	\longrightarrow \pi_0(H^k)
	\longrightarrow 0.
\end{equation}
Let $G_0$ denote the identity component of an abelian Lie group $G$. Then the real span $\Gamma^k_\R$ of $\Gamma^k$ can be identified with the Lie algebra of the maximal compact subgroup of $H_0^k$. In particular, $H_0^k$ is compact if and only if $\Gamma^k$ is a full lattice in $\fh^k$.

We set
\[
H^{\mathrm{even}} \coloneq \prod_{k\in2\Z}H^k,
\qquad
H^{\mathrm{odd}} \coloneq \prod_{k\in2\Z+1}H^k,
\]
and use the same convention for $\fh$, $\Gamma$, and subgroups of $H$.

\begin{defn}\label{def:abelian-BV-group}
A \emph{BV structure} on $H$ is a triple $(\omega,\kappa,w)$, where:
\begin{itemize}[leftmargin=2em]
\item $\omega$ is a $(-1)$-shifted symplectic pairing on $\fh$;
\item $\kappa = \bigoplus \kappa^k$ is a direct sum of inner products $\kappa^k$ on $\fh^k$ for which
\[
\omega^\flat:(\fh,\kappa)\xlongrightarrow{\sim}(\fh^\vee[-1],\kappa^\vee[-1])
\]
is an isometry;
\item $w$ is a complex quadratic weight on $\pi_0(H^0)$. 
\end{itemize}
We call $H$, equipped with such a BV structure, an \emph{abelian BV group}.
\end{defn}

We call an abelian BV group $H$ \emph{integral} if its exponential lattice $\Gamma$ is integral with respect to the $(-1)$-shifted symplectic pairing $\omega$. An integral abelian BV group is called \emph{integrable} if $\Gamma$ is, in addition, coisotropic.

\begin{defn}
	A \emph{Lagrangian} of $H$ is a family of closed subgroups
	\[
	L=\{L^k\subset H^k\}_{k\in\Z}
	\]
	such that the total Lie algebra $\fl\subset\fh$ is Lagrangian and the induced map
	\[
	\pi_0(L^k)\longrightarrow\pi_0(H^k)
	\]
	is injective for all $k\in\Z$.
\end{defn}

A Lagrangian $L$ is called \emph{admissible} if each $L_0^k$ is compact, $\pi_0(L^k)$ is finite for $k\neq 0$, and
\[
\Re(w)\big|_{\pi_0(L^0)}
\]
descends to a positive-definite quadratic form on $\pi_0(L^0)/\Tor(\pi_0(L^0))$. 

\begin{lem}
	Every integrable abelian BV group $H$ has an admissible Lagrangian.
\end{lem}

\begin{proof}
	Choose a Lagrangian subspace $\underline{\fl}$ of the reduced symplectic space $\Gamma_\R/\rad(\Gamma)_\R$ rational with respect to $\underline\Gamma = \Gamma/\rad(\Gamma)$, and let $\fl \subset\Gamma_\R$ be its inverse
	image. Then $\fl$ is Lagrangian in $\fh$ and $\Gamma\cap\fl$ is a full lattice
	in $\fl$. Hence
	\[
	L \coloneq \exp(\fl)\subset H_0
	\]
	is a compact connected Lagrangian. In particular, $\pi_0(L)=0$, so $L$ is admissible.
\end{proof}

Henceforth, we assume that $H$ is integrable.

\begin{lem}\label{lem:integrable-admissible}
	A Lagrangian $L$ is admissible if and only if
	\[
	\int_{L^{\mathrm{even}}} \left|\exp(-2\pi w)\right|\,\mu_{\mathrm{even}}<\infty,
	\qquad
	\int_{L^{\mathrm{odd}}}1\ \mu_{\mathrm{odd}}<\infty,
	\]
	where $\mu_{\mathrm{even}}$ and $\mu_{\mathrm{odd}}$ denote the positive densities on $L^{\mathrm{even}}$ and $L^{\mathrm{odd}}$ induced by $\kappa$.
\end{lem}

\begin{proof}
	Since $|\exp(-2\pi w)|=\exp(-2\pi\Re(w))$ is constant on connected components and $\mu_{\mathrm{even}}$ is translation-invariant,
	\[
	\int_{L^{\mathrm{even}}} \left|\exp(-2\pi w)\right|\,\mu_{\mathrm{even}}
	=
	\vol(L^{\mathrm{even}}_0)
	|\pi_0(L^{\mathrm{even}})/\pi_0(L^0)|
	\sum_{c\in\pi_0(L^0)}\exp(-2\pi\Re(w(c))).
	\]
	The volume factor is finite precisely when $L^{\mathrm{even}}_0$ is compact. The second factor is finite precisely when $\pi_0(L^{\mathrm{even}})/\pi_0(L^0)$ is finite. By Lemma~\ref{lem:real-quadratic}, the series is a finite torsion factor times a Gaussian sum over the free lattice. It converges if and only if the induced real quadratic form is positive definite.
		
	Likewise,
	\[
	\int_{L^{\mathrm{odd}}}1\ \mu_{\mathrm{odd}} = |\pi_0(L^{\mathrm{odd}})|
	\vol(L^{\mathrm{odd}}_0),
	\]
	which is finite precisely when the identity component $L^{\mathrm{odd}}_0$ is compact and the component group $\pi_0(L^{\mathrm{odd}})$ is finite.
\end{proof}

An admissible Lagrangian is called  \emph{maximal} if it is not properly contained in another admissible Lagrangian.

\begin{defn}
For a maximal admissible Lagrangian $L$, define
\begin{equation}\label{eq:BV-partition}
Z_H(L)\coloneq
\frac{\displaystyle\int_{L^{\mathrm{even}}}\exp(- 2\pi w)\,\mu_{\mathrm{even}}}
{\displaystyle\int_{L^{\mathrm{odd}}}1\ \mu_{\mathrm{odd}}}.
\end{equation}
We call $Z_H(L)$ the \emph{zero-mode partition function} of $H$ (with respect to $L$).
\end{defn}

The definition is motivated by two elementary features of the field-theoretic partition function. First, it is modeled on finite-dimensional BV integration over a Lagrangian. Second, partition functions are multiplicative under disjoint unions. Since the zero modes of a disjoint union form the product of the corresponding abelian BV groups, the quotient of the even and odd integrals in \eqref{eq:BV-partition} has the expected multiplicativity. 

For convenience, we set
\[
T_k\coloneq\bigl|\Tor (\pi_0(H^k))\bigr|,
\qquad
T(H)\coloneq\prod_{k\ne0}T_k^{(-1)^k},
\]
\[
\Theta(L)\coloneq\sum_{c\in\pi_0(L^0)}\exp(-2 \pi w(c))
\qquad
\vol(L_0)\coloneq\prod_k \vol(L_0^k)^{(-1)^k},
\]
with the convention $\vol(\{0\})=1$.

\begin{prop}\label{prop:partition-factorization}
For every maximal admissible Lagrangian $L$,
\begin{equation}\label{eq:partition-factorization}
Z_H(L)=\vol(L_0)\, T(H)\, \Theta(L).
\end{equation}
\end{prop}

\begin{proof}
Enlarging a Lagrangian by discrete components does not change its Lie algebra. Maximality therefore forces
\[
\pi_0(L^k)=
\Tor(\pi_0(H^k)), \qquad k\ne0.
\]
The rest follows from the proof of Lemma \ref{lem:integrable-admissible}.
\end{proof}

For a Euclidean vector space $(V,\kappa)$ and a full lattice $\Lambda\subset V$, write
\[
{\det}_\Lambda \kappa\coloneq\det\bigl(\kappa(e_i,e_j)\bigr),
\qquad
\covol(\Lambda)\coloneq\sqrt{{\det}_\Lambda \kappa},
\]
where $(e_i)$ is any basis of $\Lambda$. Then
\[
\vol(V/\Lambda) = \covol(\Lambda).
\]
For graded lattices, write $\covol(\Lambda)\coloneqq\prod_k\covol(\Lambda^k)^{(-1)^k}$.

\begin{lem}\label{lem:unimodular-volume}
Suppose that $\Gamma$ is a full unimodular lattice in $\fh$, and let $\fl\subset\fh$ be a Lagrangian such that 
\[
\Gamma_L\coloneq\Gamma\cap\fl
\]
is full in $\fl$. Then
\[
\vol(L_0)=\bigl({\det}_\Gamma \kappa \bigr)^{1/4},
\qquad
{\det}_\Gamma \kappa \coloneq\prod_k\bigl({\det}_{\Gamma^k}\kappa^k\bigr)^{(-1)^k}.
\]
In particular, the graded volume $\vol(L_0)$ is independent of the embedding $L_0\hookrightarrow H_0$.
\end{lem}

\begin{proof}
Unimodularity and the Lagrangian condition give an exact sequence of graded lattices
\[
0\longrightarrow\Gamma_L\longrightarrow\Gamma
\longrightarrow\Gamma_L^\vee[-1]\longrightarrow0.
\]
Taking Berezinian lines yields
\[
\Ber_\Z(\Gamma)\cong
\Ber_\Z(\Gamma_L)\otimes\Ber_\Z(\Gamma_L^\vee[-1])
\cong\Ber_\Z(\Gamma_L)^{\otimes2}.
\]
Compatibility of $\kappa$ and $\omega$ makes this identification an isometry. Taking norms of the integral Berezinian generators $1_\Gamma \in \Ber_\Z(\Gamma)$ and $1_{\Gamma_L} \in \Ber_\Z(\Gamma_L)$ then gives
\[
({\det}_{\Gamma}\kappa)^{1/2}
=\lVert 1_{\Gamma}\rVert
=\lVert 1_{\Gamma_L}\rVert^2
={\det}_{\Gamma_L}\kappa|_{\Gamma_L}
=\vol_{\kappa}(L_0)^2.
\]
and the result follows.
\end{proof}

We call $H$ \emph{nice} if it satisfies the following additional conditions:
\begin{itemize}
	\item $\Gamma$ is reduced unimodular;
	\item $q=\Re(w)$ descends to a positive-definite quadratic form on $\pi_0(H^0)/\Tor(\pi_0(H^0))$.
\end{itemize}

\begin{thm}\label{thm:L-independence}
The zero-mode partition function $Z_H(L)$ of a nice $H$ is independent of the maximal admissible Lagrangian $L$.
\end{thm}

\begin{proof}
By Proposition~\ref{prop:partition-factorization}, $Z_H(L)=\vol(L_0)\, T(H)\, \Theta(L)$. The torsion factor $T(H)$ is independent of $L$. Since $H$ is nice, $q=\Re(w)$ is positive definite on the real span of the free part of $\pi_0(H^0)$. Maximality therefore forces
\[
\pi_0(L^0)=\pi_0(H^0).
\]
Hence $\Theta(L)$ is also independent of $L$. 

It remains only to prove that $\vol(L_0)$ is independent of $L$. 
Since $L_0$ is compact, its Lie algebra satisfies $\fl\subset\Gamma_\R$. As $\Gamma_\R^{\perp_\omega} \subset \Gamma_\R$ and $\fl$ is Lagrangian,
\[
\Gamma_\R^{\perp_\omega} \subset \fl^{\perp_\omega}=\fl.
\]
Consequently $\Gamma\cap \Gamma_\R^{\perp_\omega} \subset\Gamma_L=\Gamma\cap\fl$. Coisotropic reduction then gives an inclusion of lattices
\[
\underline\Gamma_L \coloneq \Gamma_L/(\Gamma\cap \Gamma_\R^{\perp_\omega} )
\subset
\underline\Gamma = \Gamma/\rad(\Gamma).
\]
The latter is unimodular by assumption. Equip $\Gamma_\R^{\perp_\omega}$ and $\underline{\Gamma}_\R = \Gamma_\R/\Gamma_\R^{\perp_\omega}$ with the inner products induced by $\kappa$. Graded covolumes are multiplicative for the resulting short exact sequence of lattices
\[
0 \longrightarrow \Gamma\cap \Gamma_\R^{\perp_\omega} \longrightarrow \Gamma_L \longrightarrow \underline\Gamma_L \longrightarrow 0,
\]
which is full in the short exact sequence of Euclidean vector spaces
\[
0 \longrightarrow \Gamma_\R^{\perp_\omega} \longrightarrow (\Gamma_L)_\R = \fl \longrightarrow (\underline \Gamma_L)_\R \longrightarrow 0.
\]
It follows that
\[
\vol(L_0)
=
\covol(\Gamma_L)
=
\covol(\Gamma \cap \Gamma_\R^{\perp_\omega})\,
\covol(\underline\Gamma_L).
\]
The first factor is manifestly independent of $L$, while Lemma~\ref{lem:unimodular-volume} implies that the second factor is also independent of $L$.
\end{proof}

\section{Elliptic higher abelian gauge theories}

In this section, we apply the preceding construction to quadratic elliptic $p$-form abelian gauge theories and make explicit how their global classical data enter the zero-mode partition function. 

\subsection{BV structures on zero modes}

Let $(X,g)$ be a closed oriented Riemannian $n$-manifold. We write $\check H^q(X)$ for the degree-$q$ Cheeger--Simons group of differential characters \cite{CheegerSimons}. We use the curvature and characteristic-class maps
\[
F:\check H^q(X)\longrightarrow\Omega^q_{\Z}(X),
\qquad
c:\check H^q(X)\longrightarrow H^q(X;\Z),
\]
and the standard exact sequences
\[
0\longrightarrow H^{q-1}(X;\R/\Z)\longrightarrow\check H^q(X)
\xlongrightarrow{F}\Omega^q_{\Z}(X)\longrightarrow0,
\]
\[
0\longrightarrow\Omega^{q-1}(X)/\Omega^{q-1}_{\Z}(X)
\longrightarrow\check H^q(X)\xlongrightarrow{c}H^q(X;\Z)\longrightarrow0.
\]
Here $\Omega^q_{\Z}(X)$ denotes the closed forms with integral periods. In particular, a $p$-form $\alpha$ determines a topologically trivial differential character. We also use the standard product and integration in differential cohomology without further comment; see \cite{CheegerSimons} and \cite[Appendix~A]{BelovMoore}. 

Let $S: \check H^{p+1}(X) \rightarrow \C/i\Z$ be a complex quadratic action functional. We call $S$ smooth if, for every $\check A$, the function $\alpha\mapsto S(\check A+\alpha)$ on $\Omega^p(X)$ is smooth. We assume that its first variation is represented by a map
\[
\sE_S: \check H^{p+1}(X) \rightarrow \Omega^{n-p}_{\mathrm{cl}}(X)
\]
such that
\[
\frac{d}{dt} \bigg|_{t=0} S(\check A+t\alpha)
= \epsilon_S \int_X \alpha \wedge \sE_S(\check A),
\]
where $\epsilon_S=1$ or $-i$. For a quadratic action, the linearization of $\sE_S$ is independent of the background field; the locality condition below identifies it with a differential operator $K$.

\begin{defn}
	We call $\sE_S$ the \emph{equation-of-motion} map of $S$, and say that $S$ is \emph{local} if there exists a morphism of complexes
	\[
	K\colon \mathrm{Bund}_\nabla(p+1)
	\longrightarrow
	\Omega_{\mathrm{cl}}^{n-p}
	\]
	such that the induced map on zeroth hypercohomology agrees with $\sE_S$.
\end{defn}

Note that such a morphism $K$, if it exists, has precisely one component
\[
K\colon \Omega^p \longrightarrow \Omega^{n-p},
\]
which satisfies
\[
d \circ K  =  0, \qquad K \circ d = 0.
\]
Moreover, one has
\[
\sE_S(\alpha) = \mathbb H^0(K)(\alpha) = K(\alpha)
\]
for $\alpha \in \Omega^p(X)$. Thus, $K$ is uniquely determined by $S$.

\begin{lem}\label{selfdual}
	$K$ satisfies
	\[
	\int_X  K \alpha \wedge \beta = (-1)^{p(n-p)} \int_X \alpha \wedge K \beta.
	\]
\end{lem}
\begin{proof}
	This follows from the computation
	\[
	\frac{\partial^2}{\partial s \partial t} \bigg|_{s=t=0} S(s \beta + t\alpha) = \epsilon_S \int_X \alpha \wedge \sE_S(\beta) = \epsilon_S \int_X \alpha \wedge K \beta
	\]
	and the symmetry of partial derivatives.
\end{proof}

We call the $(-1)$-shifted mapping cone
\[
\sF\coloneq\operatorname{Cone}(K)[-1]
\]
the full \emph{BV complex} of $S$. Explicitly, $\sF$ is the complex 
\[
\Z[p+1]
\longrightarrow
\Omega^0[p]
\xlongrightarrow{d}\cdots
\xlongrightarrow{d}
\Omega^p
\xlongrightarrow{K}
\Omega^{n-p}[-1]
\xlongrightarrow{d}\cdots
\xlongrightarrow{d}
\Omega^n[-p-1].
\]

Let $\sF_{\mathrm{pert}}$ denote the complex obtained from $\sF$ by deleting the integral term $\Z[p+1]$. Since each sheaf $\Omega^q$ is fine, the hypercohomology of $\sF_{\mathrm{pert}}$ is computed by the cohomology of
its complex of global sections:
\[
\mathbb H^k(X;\sF_{\mathrm{pert}})
\cong
H^k\!\bigl(\sF_{\mathrm{pert}}(X)\bigr)
=
\begin{cases}
	H^{k+p}_{\mathrm{dR}}(X), & k<0,\\[2mm]
	\ker K\big/\Omega^p_{\mathrm{ex}}(X), & k=0,\\[2mm]
	\Omega^{n-p}_{\mathrm{cl}}(X)\big/\operatorname{im}K, & k=1,\\[2mm]
	H^{k+n-p-1}_{\mathrm{dR}}(X), & k>1,
\end{cases}
\]
where $H_{\mathrm{dR}}(X)$ denotes the de Rham cohomology of $X$.
\begin{lem}
	We have
	\[
	\Omega^p_{\mathrm{cl}}(X)\subseteq \ker K,
	\qquad
	\operatorname{im}K\subseteq \Omega^{n-p}_{\mathrm{ex}}(X).
	\]
	Using harmonic representatives with respect to $g$, there are injections
	\[
	H^p_{\mathrm{dR}}(X)
	\hookrightarrow
	H^0\!\bigl(\sF_{\mathrm{pert}}(X)\bigr),
	\qquad
	H^{n-p}_{\mathrm{dR}}(X)
	\hookrightarrow
	H^1\!\bigl(\sF_{\mathrm{pert}}(X)\bigr).
	\]
\end{lem}
\begin{proof}
	The first inclusion follows directly from the Poincar\'e lemma: every global closed $p$-form is locally exact and hence lies in the kernel of $K$. The second inclusion follows from Lemma \ref{selfdual} and Poincar\'e duality. 
\end{proof}

\begin{defn}
	We call $S$ \emph{elliptic} if $\sF_{\mathrm{pert}}(X)$ is an elliptic complex.
\end{defn}

From now on, we always assume that $S$ is elliptic. In particular, $H\bigl(\sF_{\mathrm{pert}}(X)\bigr)$ is a finite-dimensional graded vector space. We also assume that, for every $k$, the image of the connecting homomorphism
\[
\delta^k:H^{k+p}(X;\Z)\longrightarrow H^k\!\bigl(\sF_{\mathrm{pert}}(X)\bigr)
\]
is discrete. This hypothesis holds in all examples considered below.

Let $D$ denote the differential of $\sF_{\mathrm{pert}}$ and $D^*$ denote its adjoint with respect to the $L^2$ inner product $\langle -, - \rangle_g$ on $\sF_{\mathrm{pert}}(X)$ induced by the Riemannian metric $g$. Hodge theory yields the orthogonal decomposition
\[
\sF_{\mathrm{pert}}(X) = H_D(X) \oplus \operatorname{im}(D) \oplus \operatorname{im}(D^*),
\]
where
\[
 H_D(X) \coloneqq \ker D \cap \ker D^*.
\]
There is a canonical isomorphism
\[
H\bigl(\sF_{\mathrm{pert}}(X)\bigr) \xlongrightarrow{\sim}  H_D(X), \qquad [\alpha] \longmapsto \alpha_h
\] 
where $\alpha_h$ denotes the harmonic representative of $[\alpha]$.

\begin{lem}\label{star}
	The Hodge star operator
	\[
	\star: H_D^k(X) \xlongrightarrow{\sim} H_D^{1-k}(X)
	\]
	defines an isometry of inner product spaces.
\end{lem}
\begin{proof}
	By Lemma \ref{selfdual}, $D$ satisfies
	\[
	\int_X D \alpha \wedge \beta = \pm \int_X \alpha \wedge D \beta.
	\]
	It follows that
	\[
	D^*=\pm \star D\star.
	\]
	Hence $H_D(X)=\ker D\cap\ker(D\star)$ and $H_D(X)$ is preserved by $\star$. Since $\star^2=\pm 1$, the Hodge star restricts to an isomorphism on $H_D(X)$. Moreover, it is an isometry with respect to the inner product:
	$
	\langle \star\alpha_h,\star\beta_h\rangle_g
	=
	\langle \alpha_h,\beta_h\rangle_g.
	$
\end{proof}

\begin{prop}
	The finite-dimensional graded vector space $H\bigl(\sF_{\mathrm{pert}}(X)\bigr)$ carries a canonical $(-1)$-shifted symplectic pairing
	\[
	\omega\colon
	H^k\!\bigl(\sF_{\mathrm{pert}}(X)\bigr)
	\times
	H^{1-k}\!\bigl(\sF_{\mathrm{pert}}(X)\bigr)
	\longrightarrow \R,
	\]
	given by
	\[
	\omega([\alpha],[\beta])
	=
	\begin{cases}
		\displaystyle \int_X \alpha_h\wedge\beta_h, & k\leq 0,\\[6pt]
		\displaystyle -\int_X \beta_h\wedge\alpha_h, & k>0,
	\end{cases}
	\]
	where $\alpha_h$ and $\beta_h$ denote the harmonic representatives of
	$[\alpha]$ and $[\beta]$, respectively.
	
	Moreover, $\omega$ is compatible with the inner product
	\[
	\kappa\colon
	H^k\!\bigl(\sF_{\mathrm{pert}}(X)\bigr)
	\times
	H^k\!\bigl(\sF_{\mathrm{pert}}(X)\bigr)
	\longrightarrow \R,
	\qquad
	\kappa([\alpha],[\beta])
	=
	\langle \alpha_h,\beta_h\rangle_g
	=
	\int_X \alpha_h\wedge\star\beta_h.
	\]
\end{prop}

\begin{proof}
	Nondegeneracy of $\omega$ follows from the Hodge decomposition. More precisely,
	if $\alpha_h\in H_D(X)$, then
	\[
	\int_X \alpha_h\wedge\bigl(D\gamma\pm\star D\star\eta\bigr)
	=
	\pm\int_X \star\alpha_h\wedge D\star\eta
	=
	\pm\int_X D\star\alpha_h\wedge\star\eta
	=
	0.
	\]
	Thus $H_D(X)$ is orthogonal, with respect to the Poincar\'e pairing, to the non-harmonic summands in the Hodge decomposition. Since the Poincar\'e pairing on differential forms is nondegenerate, its restriction to $H_D(X)$ is therefore nondegenerate.
	
	To see that $\omega$ is compatible with $\kappa$, observe that
	\[
	\omega^\flat=\pm \kappa^\flat \circ \star.
	\]
	By Lemma \ref{star}, the Hodge star is an isometry, and hence $\omega^\flat$ is an isometry.
\end{proof}

There is a short exact sequence 
\[
0\longrightarrow \sF_{\mathrm{pert}}
\longrightarrow \sF
\longrightarrow \Z[p+1]
\longrightarrow 0,
\]
which induces a long exact sequence of abelian groups
\[
\begin{tikzcd}
	\cdots \arrow[r]
	&
	H^k\bigl(\sF_{\mathrm{pert}}(X)\bigr) \arrow[r]
	&
	\mathbb H^k(X;\sF) \arrow[r]
	&
	H^{k+p+1}(X;\Z)
	\arrow[dll, "\delta"', rounded corners, to path={
		-- ([xshift=2ex]\tikztostart.east)
		|- ([yshift=-3ex]\tikztostart.south)
		-| ([xshift=-2ex]\tikztotarget.west)
		-- (\tikztotarget)
	}]
	\\
	&
	H^{k+1}\bigl(\sF_{\mathrm{pert}}(X)\bigr) \arrow[r]
	&
	\mathbb H^{k+1}(X;\sF) \arrow[r]
	&
	H^{k+p+2}(X;\Z) \arrow[r]
	&
	\cdots
\end{tikzcd}
\]
Denote by
\[
\delta^k:H^{k+p}(X;\Z)\longrightarrow H^k\!\bigl(\sF_{\mathrm{pert}}(X)\bigr)
\]
the connecting homomorphism. 

By the discreteness assumption above, $\operatorname{im}(\delta^k)$ is a lattice in $H^k\bigl(\sF_{\mathrm{pert}}(X)\bigr)$. The long exact sequence therefore naturally endows $\mathbb H^k(X;\sF)$ with the structure of an abelian Lie group.

\begin{lem}\label{lem:master-hypercohomology}
	For every $k$, there is a natural exact sequence of abelian groups
	\begin{equation}\label{eq:master-hypercohomology}
		0
		\longrightarrow \operatorname{im}(\delta^k)
		\longrightarrow H^k\bigl(\sF_{\mathrm{pert}}(X)\bigr)
		\longrightarrow \mathbb H^k(X;\sF)
		\longrightarrow \ker(\delta^{k+1})
		\longrightarrow 0.
	\end{equation}
	In particular, $\mathbb H^k(X;\sF)$ is canonically an abelian Lie group with
	Lie algebra
	\[
	H^k\bigl(\sF_{\mathrm{pert}}(X)\bigr),
	\]
	exponential lattice
	\[
	\operatorname{im}(\delta^k)
	\subset H^k\bigl(\sF_{\mathrm{pert}}(X)\bigr),
	\]
	and group of connected components
	\[
	\ker(\delta^{k+1})
	\subset H^{k+p+1}(X;\Z).
	\]
\end{lem}

There is also a short exact sequence
\[
0\longrightarrow \Omega_{\mathrm{cl}}^{n-p}[-1]
\longrightarrow \sF
\xlongrightarrow{\mathrm{For}} \mathrm{Bund}_\nabla(p+1)
\longrightarrow 0,
\]
which induces a long exact sequence in hypercohomology
\[
\begin{tikzcd}
	\cdots \arrow[r]
	&
	\mathbb H^{k-1}(X;\Omega_{\mathrm{cl}}^{n-p}) \arrow[r]
	&
	\mathbb H^k(X;\sF) \arrow[r]
	&
	\mathbb H^k(X;\mathrm{Bund}_\nabla(p+1))
	\arrow[dll, "\delta"', rounded corners, to path={
		-- ([xshift=2ex]\tikztostart.east)
		|- ([yshift=-3ex]\tikztostart.south)
		-| ([xshift=-2ex]\tikztotarget.west)
		-- (\tikztotarget)
	}]
	\\
	&
	\mathbb H^{k}(X;\Omega_{\mathrm{cl}}^{n-p}) \arrow[r]
	&
	\mathbb H^{k+1}(X;\sF) \arrow[r]
	&
	\mathbb H^{k+1}(X;\mathrm{Bund}_\nabla(p+1)) \arrow[r]
	&
	\cdots
\end{tikzcd}
\]
In particular, around degree zero we have
\[
0  \longrightarrow \mathbb H^0(X;\sF) \xlongrightarrow{\mathbb{H}^0(\mathrm{For})} 
\check{H}^{p+1}(X) \xlongrightarrow{\mathbb{H}^0(K)} \Omega^{n-p}_{\mathrm{cl}}(X) \longrightarrow \cdots
\]
We then obtain an injective group homomorphism
\[
\mathbb{H}^0(\mathrm{For}): \mathbb H^0(X; \sF) \hookrightarrow \check{H}^{p+1}(X).
\]
We define
\[
S_{\mathrm{crit}}
\coloneqq
S\circ\mathbb H^0(\mathrm{For})
\colon
\mathbb H^0(X;\sF)
\longrightarrow
\C/ i\Z
\]
as the restriction of \(S\) to the critical locus \(\mathbb H^0(X;\sF)\).

\begin{lem}
	The function $S_{\mathrm{crit}}$ is constant on every connected component of $\mathbb H^0(X;\sF)$ and vanishes on the identity component. Thus, it descends to a complex quadratic weight
	\[
	w \colon \pi_0 \bigl(\mathbb H^0(X;\sF)\bigr)
	\longrightarrow \C/i\Z.
	\]
\end{lem}

\begin{proof}
	The identity component of $\mathbb H^0(X;\sF)$ is the image of
	\[
	\exp\colon
	H^0(\sF_{\mathrm{pert}}(X))
	=
	\frac{\ker K}{\Omega_{\mathrm{ex}}^p(X)}
	\longrightarrow
	\mathbb H^0(X;\sF).
	\]
	Let $x\in\mathbb H^0(X;\sF)$ and $[\alpha]\in H^0(\sF_{\mathrm{pert}}(X))$, with harmonic representative $\alpha_h$. The path
	\[
	t\longmapsto x+\exp(t[\alpha])
	\]
	lies in the same connected component of the critical locus. Writing $\check A=\mathbb H^0(\mathrm{For})(x)$, linearity of the equation-of-motion map gives
	\[
	\sE_S(\check A+t\alpha_h)=\sE_S(\check A)+tK\alpha_h = \mathbb H^0(K)(\check A) = 0,
	\]
	where we use $K \alpha_h = 0$ and $\mathbb H^0(K) \circ\mathbb H^0(\mathrm{For}) = 0$. It follows that
	\[
	\frac{d}{dt}S_{\mathrm{crit}}\bigl(x+\exp(t[\alpha])\bigr)
	=\epsilon_S\int_X\alpha_h\wedge\sE_S(\check A+t\alpha_h) = 0.
	\]
	Thus $S_{\mathrm{crit}}$ is constant on every connected component. On the identity component this constant is $S(0)=0$.
\end{proof}

We have thus proved the following.

\begin{thm}
	The data $(\omega,\kappa,w)$ defined above endow the collection of abelian Lie groups $\mathbb H^\bullet(X;\sF)$ with the structure of an abelian BV group.
\end{thm}

For later use, let
\[
\rho^r:H^r(X;\Z)\longrightarrow H^r_{\mathrm{dR}}(X)
\]
denote the integral-to-de Rham map, and put
\[
\sH^r \coloneq \ker \Delta^r,
\qquad
\Lambda^r \coloneq \operatorname{im}\bigl(H^r(X;\Z)\xlongrightarrow{\rho^r} H^r_{\mathrm{dR}}(X) \xlongrightarrow{\sim} \sH^r\bigr),
\qquad
\tau_r\coloneq|\Tor H^r(X;\Z)|,
\]
where $\Delta^r$ is the Hodge Laplacian on $\Omega^r(X)$. 

\subsection{Maxwell theory}

Let $F_{\check A}$ denote the curvature of $\check{A} \in \check{H}^{p+1}(X)$. For $e>0$, the Euclidean Maxwell action is
\[
S_{\mathrm{M}}(\check A)
=
\frac{1}{2 e^2}
\int_X F_{\check A}\wedge \star F_{\check A}.
\]
With $\epsilon_{S_{\mathrm{M}}}=(-1)^{p+1}$, the middle operator $K$ is 
\[
K_{\mathrm M}= e^{-2}\,d\star d.
\]
The BV complex $\sF_{\mathrm M}$ is
\begin{equation*}
\Z[p+1]\rightarrow
\Omega^0[p]\xlongrightarrow d\cdots\xlongrightarrow d
\Omega^p
\xlongrightarrow{e^{-2}d\star d}
\Omega^{n-p}[-1]\xlongrightarrow d\cdots\xlongrightarrow d
\Omega^n[-p-1].
\end{equation*}

\begin{lem}
The perturbative hypercohomology is
\[
H^k(\sF_{\mathrm M, \mathrm{pert}}(X)) =
\begin{cases}
H^{p+k}_{\mathrm{dR}}(X),&k \leq 0,\\
H^{n-p+k-1}_{\mathrm{dR}}(X),&k >0.
\end{cases}
\]
Under these identifications,
\[
\delta^k=
\begin{cases}
\rho^{p+k},&k\le0,\\
0,&k>0,
\end{cases}
\]

\end{lem}

\begin{proof}
Away from the middle arrow, the perturbative BV complex is the de Rham complex. At degree zero,
\[
H^0(\sF_{\mathrm M, \mathrm{pert}}(X))
=\ker K_{\mathrm M}/\Omega^p_{\mathrm{ex}}(X).
\]
For the Maxwell operator, Hodge decomposition gives $\ker K_{\mathrm M}=\Omega^p_{\mathrm{cl}}(X)$ and $\operatorname{im} K_{\mathrm M} = \Omega^{n-p}_{\mathrm{ex}}(X)$; hence
\[
H^0(\sF_{\mathrm M, \mathrm{pert}}(X))= H^p_{\mathrm{dR}}(X), \qquad H^1(\sF_{\mathrm M, \mathrm{pert}}(X))= H^{n-p}_{\mathrm{dR}}(X).
\]

For $k \leq 0$, the connecting map is induced by $\Z\hookrightarrow\R$ and is therefore the usual integral-to-de Rham map. For $k > 0$, a topological class contributes through the curvature, but the equation-of-motion term is $d \star F$ and is exact; its cohomology class is therefore zero. 
\end{proof}

\begin{rmk}
	Therefore, for $j<0$ there is a natural short exact sequence
	\[
	0\longrightarrow
	\frac{\mathcal H^{p+j}}{\Lambda^{p+j}}
	\longrightarrow
	\Hq^j(X;\sF_{\mathrm M})
	\longrightarrow
	\Tor(H^{p+j+1}(X;\Z))
	\longrightarrow 0,
	\]
	while in degree zero,
	\[
	0\longrightarrow
	\frac{\mathcal H^p}{\Lambda^p}
	\longrightarrow
	\Hq^0(X;\sF_{\mathrm M})
	\longrightarrow
	H^{p+1}(X;\Z)
	\longrightarrow 0.
	\]
	For $j>0$, one has
	\[
	0\longrightarrow
	\mathcal H^{n-p+j-1}
	\longrightarrow
	\Hq^j(X;\sF_{\mathrm M})
	\longrightarrow
	H^{p+j+1}(X;\Z)
	\longrightarrow 0.
	\]
\end{rmk}

Set
\[
H_{\mathrm M}^k\coloneq\mathbb H^k(X;\sF_{\mathrm M}),
\qquad H_{\mathrm M}\coloneq\{H_{\mathrm M}^k\}_{k\in\Z}.
\]
The total Lie algebra of the zero-mode abelian BV group $H_{\mathrm M}$ is therefore
\[
\fh_{\mathrm M}^k\cong
\begin{cases}
\mathcal H^{p+k},&k \leq 0,\\
\mathcal H^{n-p+k-1},&k >0,
\end{cases}
\]
with inner product $\kappa_{\mathrm M}$ given by the $L^2$ inner product
\[
\kappa_{\mathrm M}(\alpha,\beta)= \int_X\alpha\wedge \star \beta,
\qquad |\alpha|=|\beta|,
\]
and shifted symplectic pairing induced by Poincar\'e duality,
\[
\omega_{\mathrm M}(\alpha,\beta)= \pm \int_X\alpha\wedge\beta,
\qquad |\alpha|+|\beta|=1.
\]
The exponential lattice is
\[
\Gamma_{\mathrm M}^k=
\begin{cases}
\Lambda^{p+k},& k \leq 0,\\
0,&k >0,
\end{cases}
\]
so $\Gamma_{\mathrm M}$ is integral and Lagrangian in $\fh_{\mathrm M}$, hence coisotropic and reduced unimodular, with trivial reduction. Finally,
\[
\pi_0(\Hq^0(X;\sF_{\mathrm M})) \cong H^{p+1}(X;\Z),
\]
and for $c\in H^{p+1}(X;\Z)$ the zero-mode weight is
\[
w(c)
=
\frac{1}{2e^2}\int_X c_\R \wedge \star c_\R,
\]
where $c_\R\in\mathcal H^{p+1}$ is the harmonic representative of $c$ modulo torsion. 

In summary, we obtain the following.
\begin{prop}
	The abelian BV group $H_{\mathrm M}$ is nice.
\end{prop}

Any maximal admissible Lagrangian $L\subset H_{\mathrm M}$ has Lie algebra $\Gamma_{\mathbb R}$ and component groups
\[
\pi_0(L^k)=
\begin{cases}
	\Tor (H^{k+p+1}(X;\Z)), & k\neq 0,\\[2mm]
	H^{p+1}(X;\Z), & k=0.
\end{cases}
\]
The three factors in Proposition~\ref{prop:partition-factorization} are now completely explicit: the free part of $H^{p+1}(X;\Z)$ gives the Gaussian theta sum, the torsion component groups give the alternating finite factor, and the compact harmonic tori give the graded volume. Thus
\begin{equation}\label{eq:maxwell-partition}
Z_{H_{\mathrm{M}}}(L)
=
\sum_{c\in H^{p+1}(X;\Z)_{\mathrm{free}}}
\exp\!\left(-\frac{\pi}{e^2}
\int_X c_\R \wedge \star c_\R
\right)
\prod_{r=-p}^{n-p-1}
\tau_{r+p+1}^{(-1)^{r}}
\prod_{-p \leq r \leq 0} \vol\left(\frac{\sH^{r+p}}{\Lambda^{r+p}}\right)^{(-1)^r}.
\end{equation}
By Theorem~\ref{thm:L-independence}, this expression is independent of the maximal admissible Lagrangian $L$.

\begin{rmk}\label{rmk:maxwell-comparison}
	Compared with the zero-mode contribution to the higher Maxwell partition function obtained in \cite[Eq.~(22)]{DonnellyMichelWall} and \cite[Eq.~(47.21)]{MooreSaxena}, \eqref{eq:maxwell-partition} contains a different torsion factor. This difference arises essentially from the presence of antifields in the BV complex together with the maximality requirement on the gauge-fixing Lagrangian $L$.
\end{rmk}

\subsection{Abelian Chern--Simons theory}\label{sec: CS}

Let $p=2\ell+1$ and $n=2p+1=4\ell+3$. Let $k\in\Z\setminus\{0\}$. The higher abelian Chern--Simons theory is defined by an imaginary quadratic action on $\check H^{p+1}(X)$. Formally, one would like to write
\[
S_{\mathrm{CS}}(\check A)
=
-\frac{ik}{2}\int_X \check A\cup\check A,
\qquad
\check A\in\check H^{p+1}(X).
\]
However, the differential-cohomology pairing takes values in $\R/\Z$, where division by $2$ is not canonically defined. The half-square must therefore be replaced by a quadratic refinement of the cup-product pairing.

For that purpose, we fix a ``Wu-cocycle" $\eta\in Z^p(X;\Q/\Z)$ in the sense of Hopkins--Singer \cite{HopkinsSinger}, together with a compatible differential-cohomological quadratic refinement
\[
Q_\eta:\check H^{p+1}(X)\longrightarrow\R/\Z
\]
satisfying
\[
Q_\eta(\check A+\check B)-Q_\eta(\check A)-Q_\eta(\check B)
=\int_X\check A\cup\check B,
\]
and
\[
\left.\frac{d}{dt}\right|_{t=0}Q_\eta(\check A+t\alpha)
=\int_X\alpha\wedge F_{\check A};
\]
see \cite[\S4 and Appendix~A]{BelovMoore}. On flat torsion fields it induces a quadratic refinement
\[
q_\eta:\Tor \left(H^{p+1}(X;\Z)\right) \longrightarrow\Q/\Z
\]
of the torsion linking pairing associated with the Wu-cocycle $\eta$ \cite[Prop.~5.66]{HopkinsSinger}.

Now set
\[
S_{\mathrm{CS},\eta}(\check A)=-ikQ_\eta(\check A).
\]
With $\epsilon_{S_{\mathrm{CS}}}=-i$, the middle operator is
\[
K_{\mathrm{CS}}=kd,
\]
and the BV complex $\sF_{\mathrm{CS}}$ is
\begin{equation*}
\Z[p+1]\longrightarrow
\Omega^0[p]\xlongrightarrow d\cdots\xlongrightarrow d
\Omega^p\xlongrightarrow{k d}
\Omega^{p+1}[-1]\xlongrightarrow d\cdots\xlongrightarrow d
\Omega^{2p+1}[-p-1].
\end{equation*}

\begin{lem}\label{lem:CS-cohomology}
The perturbative hypercohomology and connecting maps are
\[
H^j(\sF_{\mathrm{CS},\mathrm{pert}}(X))\cong H^{p+j}_{\mathrm{dR}}(X),
\qquad
\delta^j=a_j\rho^{p+j},
\qquad
a_j=\begin{cases}1,&j\le0,\\ k,&j>0.\end{cases}
\]
\end{lem}
\begin{proof}
The perturbative complex differs from the de Rham complex only by multiplying its middle differential by the nonzero scalar $k$, so its cohomology is unchanged. The connecting map is the integral-to-de Rham map before the middle arrow and acquires one factor of $k$ after crossing it.
\end{proof}

Lemma~\ref{lem:master-hypercohomology} therefore gives
\begin{equation}\label{eq:CS-H}
0\longrightarrow
\frac{\mathcal H^{p+j}}{a_j\Lambda^{p+j}}
\longrightarrow
\Hq^j(X;\sF_{\mathrm{CS}})
\longrightarrow
\Tor \left( H^{p+j+1}(X;\Z)\right)
\longrightarrow0.
\end{equation}
In particular,
\[
\pi_0(\Hq^0(X;\sF_{\mathrm{CS}}))\cong\Tor \left( H^{p+1}(X;\Z) \right),
\]
so only flat torsion sectors survive the equation of motion.

Set
\[
H_{\mathrm{CS}}^j\coloneq\mathbb H^j(X;\sF_{\mathrm{CS}}),
\qquad H_{\mathrm{CS}}\coloneq\{H_{\mathrm{CS}}^j\}_{j\in\Z}.
\]
Its Lie algebra is $\fh_{\mathrm{CS}}^j\cong\mathcal H^{p+j}$, with the $L^2$ metric and Poincar\'e pairing as above, while
\[
\Gamma_{\mathrm{CS}}^j=a_j\Lambda^{p+j}.
\]
$\Gamma_{\mathrm{CS}}$ is full and integral, and is unimodular exactly when $k = \pm 1$. The restriction of $S_{\mathrm{CS},\eta}$ to flat torsion characters gives 
\[
w_{\mathrm{CS}}(t)=-ikq_\eta(t).
\]
The preceding discussion yields the following.
\begin{prop}
	$H_{\mathrm{CS}}$ is integrable, and is nice precisely for $k=\pm1$.
\end{prop}

For $k=\pm1$, Proposition~\ref{prop:partition-factorization} and Lemma~\ref{lem:unimodular-volume} give
\begin{equation}\label{eq:CS-partition}
Z_{H_{\mathrm{CS}}}(L)
=
\left(\sum_{t\in\Tor \left(H^{p+1}(X;\Z)\right)}e^{2\pi i k q_\eta(t)}\right)
\prod_{\substack{-p\leq r\leq p+1\\r\neq0}}
\tau_{r+p+1}^{(-1)^r}
\prod_{r=0}^{n}\vol(\sH^r/\Lambda^r)^{(-1)^{r-p}/2}.
\end{equation}
Again, the first two factors are respectively $\Theta(L)$ and $T(H_{\mathrm{CS}})$, while the last follows from
$\vol(L_0)=({\det}_{\Gamma_{\mathrm{CS}}}\kappa_{\mathrm{CS}})^{1/4}$ and Poincar\'e duality. The result is independent of the maximal admissible Lagrangian $L$.

\subsection{Maxwell--Chern--Simons theory}

Assume again that $p=2\ell+1$ and $n=2p+1$, with $e>0$ and $k\in\Z\setminus\{0\}$. We consider the Maxwell--Chern--Simons action on the Riemannian manifold $X$, with a purely imaginary Maxwell term:
\[
S_{\mathrm{MCS}}(\check A)
=-ikQ_\eta(\check A)-\frac{i}{2e^2}\int_XF_{\check A}\wedge\star F_{\check A}.
\]
With $\epsilon_{S_\mathrm{MCS}}=-i$, the middle operator is
\[
K_{\mathrm{MCS}}=d(e^{-2}\star d+k),
\]
and the BV complex $\sF_{\mathrm{MCS}}$ is
\begin{equation*}
	\Z[p+1]\longrightarrow
	\Omega^0[p]\xlongrightarrow d\cdots\xlongrightarrow d
	\Omega^p\xlongrightarrow{k d + e^{-2} d \star d}
	\Omega^{p+1}[-1]\xlongrightarrow d\cdots\xlongrightarrow d
	\Omega^{2p+1}[-p-1].
\end{equation*}
Set
\[
R_{e,k}=\ker\bigl(e^{-2}\star d+k:\Omega^p_{\mathrm{coex}}(X)\to\Omega^p_{\mathrm{coex}}(X)\bigr).
\]
Note that $R_{e,k}$ is finite-dimensional; it is nontrivial precisely when $-e^2k$ lies in the spectrum of $\star d$ on $\Omega^p_{\mathrm{coex}}(X)$.

\begin{lem}\label{lem:MCS-cohomology}
The perturbative cohomology agrees with that of Chern--Simons theory away from the middle degrees, while
\[
H^0(\sF_{\mathrm{MCS},\mathrm{pert}}(X))
\cong H^p_{\mathrm{dR}}(X)\oplus R_{e,k},
\qquad
H^1(\sF_{\mathrm{MCS},\mathrm{pert}}(X))
\cong H^{p+1}_{\mathrm{dR}}(X)\oplus dR_{e,k}.
\]
\end{lem}
\begin{proof}
Under the Hodge decomposition
$\Omega^p=\mathcal H^p\oplus\Omega^p_{\mathrm{ex}}\oplus\Omega^p_{\mathrm{coex}}$,
the middle differential vanishes on the first two summands and equals $d(e^{-2}\star d+k)$ on the last. Since $d$ is injective on coexact forms,
\[
\ker K_{\mathrm{MCS}}=\Omega^p_{\mathrm{cl}}\oplus R_{e,k}.
\]
Self-adjointness gives $\Omega^p_{\mathrm{coex}}=R_{e,k}\oplus\operatorname{im}(e^{-2}\star d+k)$; applying $d$ and taking the quotient in degree one gives the second formula.
\end{proof}

Set
\[
H_{\mathrm{MCS}}^j\coloneq\mathbb H^j(X;\sF_{\mathrm{MCS}}),
\qquad H_{\mathrm{MCS}}\coloneq\{H_{\mathrm{MCS}}^j\}_{j\in\Z}.
\]
If $R_{e,k}=0$, then the perturbative complex $\sF_{\mathrm{MCS},\mathrm{pert}}$ is quasi-isomorphic to the Chern--Simons perturbative complex $\sF_{\mathrm{CS},\mathrm{pert}}$. The connecting maps also agree under the harmonic identifications, since the Maxwell contribution $e^{-2}d\star F$ is exact. The full hypercohomology groups therefore fit into the exact sequences in Section~\ref{sec: CS}, with the same Lie algebras, exponential lattices, shifted symplectic pairings, and inner products. In particular,
\[
\pi_0(\Hq^0(X;\sF_{\mathrm{MCS}}))
\cong\Tor(H^{p+1}(X;\Z)).
\]
In the absence of resonance, critical fields in torsion sectors are flat, so the Maxwell term vanishes and the degree-zero weight equals the Chern--Simons weight. 
\begin{prop}
	If $R_{e,k}=0$, then
	\begin{equation*}
		Z_{H_{\mathrm{MCS}}}(L)=Z_{H_{\mathrm{CS}}}(L),
	\end{equation*}
	with the right-hand side given by \eqref{eq:CS-partition}.
\end{prop}
This equality concerns only the zero-mode contribution. In dimension $3$, a factorization of the full partition function is derived in \cite[Eq.~(14)]{Armoni2023}, in the conventions of that paper:
\[
Z_{\mathrm{MCS}}=Z_{\mathrm{DJ}}Z_{\mathrm{CS}},
\]
where $Z_{\mathrm{DJ}}$ is the massive Deser--Jackiw sector \cite{DeserJackiw} and $Z_{\mathrm{CS}}$ is the topological Chern--Simons sector. In the present model, $R_{e,k}=0$ means that there are no additional zero modes. If $R_{e,k}\neq0$, resonant zero modes appear in the massive Deser--Jackiw sector without corresponding exponential lattice directions. Consequently, $H_{\mathrm{MCS}}$ is not integrable in the sense of Section~2. We expect a density-valued extension of the zero-mode partition function associated with the resonant zero modes; its construction is beyond the scope of this paper.

\subsection{Further examples}

One may replace the gauge group $\U(1)$ by an arbitrary compact connected abelian group and allow collections of higher-form gauge fields of different degrees with quadratic couplings between them. This generalization includes, for example, abelian $BF$ theory and can be treated within a natural extension of the present framework.

\renewcommand{\refname}{References}


\begin{bibdiv}
	\begin{biblist}
		
		\bib{Armoni2023}{article}{
			author={Armoni, Adi},
			title={$S$-dual of Maxwell--Chern--Simons theory},
			journal={Phys. Rev. Lett.},
			volume={130},
			date={2023},
			number={14},
			pages={141601},
		}
		
		\bib{BeckerBeniniSchenkelSzabo}{article}{
			author={Becker, Christian},
			author={Benini, Marco},
			author={Schenkel, Alexander},
			author={Szabo, Richard J.},
			title={Abelian duality on globally hyperbolic spacetimes},
			journal={Comm. Math. Phys.},
			volume={349},
			date={2017},
			number={1},
			pages={361--392},
		}
		
		\bib{BelovMoore}{misc}{
			author={Belov, Dmitriy M.},
			author={Moore, Gregory W.},
			title={Holographic action for the self-dual field},
			date={2006},
			eprint={hep-th/0605038},
		}
		
		\bib{Brylinski}{book}{
			author={Brylinski, Jean-Luc},
			title={Loop Spaces, Characteristic Classes and Geometric Quantization},
			series={Progress in Mathematics},
			volume={107},
			publisher={Birkh\"auser Boston},
			place={Boston, MA},
			date={1993},
		}
		
		\bib{CheegerSimons}{article}{
			author={Cheeger, Jeff},
			author={Simons, James},
			title={Differential characters and geometric invariants},
			date={1985},
			conference={
				title={Geometry and Topology},
				address={College Park, MD},
				date={1983/1984},
			},
			book={
				series={Lecture Notes in Mathematics},
				volume={1167},
				publisher={Springer},
				place={Berlin},
			},
			pages={50--80},
		}
		
		\bib{DeserJackiw}{article}{
			author={Deser, Stanley},
			author={Jackiw, Roman},
			title={``Self-duality'' of topologically massive gauge theories},
			journal={Phys. Lett. B},
			volume={139},
			date={1984},
			number={5--6},
			pages={371--373},
		}
		
		\bib{DonnellyMichelWall}{article}{
			author={Donnelly, William},
			author={Michel, Ben},
			author={Wall, Aron C.},
			title={Electromagnetic duality and entanglement anomalies},
			journal={Phys. Rev. D},
			volume={96},
			date={2017},
			number={4},
			pages={045008},
		}
		
		\bib{ElliottGwilliamSaberiWilliams}{misc}{
			author={Elliott, Chris},
			author={Gwilliam, Owen},
			author={Saberi, Ingmar},
			author={Williams, Brian R.},
			title={Duality of generalized Maxwell theories as an equivalence
				in derived geometry},
			date={2026},
			eprint={2603.19161},
		}
		
		\bib{FreedMooreSegal}{article}{
			author={Freed, Daniel S.},
			author={Moore, Gregory W.},
			author={Segal, Graeme},
			title={Heisenberg groups and noncommutative fluxes},
			journal={Ann. Physics},
			volume={322},
			date={2007},
			number={1},
			pages={236--285},
		}
		
		\bib{HopkinsSinger}{article}{
			author={Hopkins, Michael J.},
			author={Singer, Isadore M.},
			title={Quadratic functions in geometry, topology, and $M$-theory},
			journal={J. Differential Geom.},
			volume={70},
			date={2005},
			number={3},
			pages={329--452},
		}
		
		\bib{JiangLiZabzine}{misc}{
			author={Jiang, Shuhan},
			author={Li, Si},
			author={Zabzine, Maxim},
			title={Abelian duality in the {BV} formalism},
			note={In preparation},
		}
		
		\bib{MooreSaxena}{article}{
			author={Moore, Gregory W.},
			author={Saxena, Vivek},
			title={{TASI} lectures on topological field theories and
				differential cohomology},
			date={2025},
			note={With an appendix by Daniel S. Freed},
			eprint={2510.07408},
		}
		
		\bib{RaySinger1971}{article}{
			author={Ray, Daniel B.},
			author={Singer, Isadore M.},
			title={$R$-torsion and the Laplacian on Riemannian manifolds},
			journal={Adv. Math.},
			volume={7},
			date={1971},
			number={2},
			pages={145--210},
		}
		
		\bib{Schwarz1978}{article}{
			author={Schwarz, Albert S.},
			title={The partition function of degenerate quadratic functional and {Ray--Singer} invariants},
			journal={Lett. Math. Phys.},
			volume={2},
			date={1978},
			number={3},
			pages={247--252},
		}
		
		\bib{Szabo2012}{article}{
			author={Szabo, Richard J.},
			title={Quantization of higher abelian gauge theory in generalized differential cohomology},
			journal={PoS},
			volume={Rio de Janeiro 2012},
			date={2012},
			pages={009},
			eprint={1209.2530},
		}
		
	\end{biblist}
\end{bibdiv}
\end{document}